\documentclass{tCON2e}

\usepackage{epstopdf}% To incorporate .eps illustrations using PDFLaTeX, etc.
\usepackage{subfigure}% Support for small, `sub' figures and tables

\usepackage[longnamesfirst,sort]{natbib}% Citation support using natbib.sty
\bibpunct[, ]{(}{)}{;}{a}{,}{,}% Citation support using natbib.sty
\renewcommand\bibfont{\fontsize{10}{12}\selectfont}% To set the list of references in 10 point font using natbib.sty

\theoremstyle{plain}% Theorem-like structures
\newtheorem{theorem}{Theorem}

\newtheorem{corollary}{Corollary}

\newtheorem{proposition}{Proposition}

\theoremstyle{definition}

\theoremstyle{remark}
\newtheorem{remark}{Remark}

\usepackage{tikz}
\usetikzlibrary{shapes,arrows,shadows,decorations.pathmorphing,calc}

\tikzstyle{block} = [draw, rectangle,  thick, fill=white, minimum width=13mm, minimum height=5mm]
\tikzstyle{sum} = [draw, circle,inner sep=0mm, minimum size=1.5mm]
\tikzstyle{connector} = [->, thick, >=latex]

\renewcommand{\vec}[1]{\bm{#1}}
\newcommand{\rank}[1]{\textnormal{rank}\left(#1\right)}
\renewcommand{\dim}[1]{\textnormal{dim}\left(#1\right)}
\renewcommand{\det}[1]{\textnormal{det}\left(#1\right)}

\newcommand{\RZ}{\mathbb{R}}

\def \zero {\vec{0}}

\def \x {\vec{x}}
\def \y {\vec{y}}
\def \u {\vec{u}}
\def \v {\vec{v}}
\def \z {\vec{z}}
\def \w {\vec{w}}
\def \d {\vec{d}}

\def \e {\vec{e}}
\def \f {\vec{f}}

\def \t {\vec{t}}
\def \o {\vec{\omega}}
\def \vpi {\vec{\pi}}

\def \A {\vec{A}}
\def \B {\vec{B}}
\def \C {\vec{C}}
\def \D {\vec{D}}

\def \F {\vec{F}}

\def \I {\vec{I}}

\def \K {\vec{K}}

\def \M {\vec{M}}
\def \N {\vec{N}}
\def \P {\vec{P}}
\def \Q {\vec{Q}}
\def \R {\vec{R}}
\def \S {\vec{S}}
\def \T {\vec{T}}

\def \V {\vec{V}}

\def \Z {\vec{Z}}

\def \mPi {\vec{\Pi}}
\def \mGamma {\vec{\Gamma}}
\def \mLambda {\vec{\Lambda}}
\def \mXi {\vec{\Xi}}
\def \mPsi {\vec{\Psi}}
\def \mPhi {\vec{\Phi}}
\def \mTheta {\vec{\Theta}}

\newcommand{\tp}{{\textnormal{\textsf{T}}}}

\begin{document}

%\jvol{00} \jnum{00} \jyear{2015} \jmonth{February}

% \articletype{GUIDE}

\title{From output regulation theory to flatness based tracking: \break a bridge for linear systems}

% \maketitle  %for concealing the authors' names

\author{
% \name{Saman Khodaverdian\textsuperscript{}$^{\ast}$\thanks{$^\ast$Email: saman.khodaverdian@rmr.tu-darmstadt.de}}
\name{Saman Khodaverdian}
\affil{\textsuperscript{}Institute of Automatic Control and Mechatronics, Control Methods and Robotics Lab, Technische Universit\"at Darmstadt, Landgraf-Georg-Str. 4, 64283 Darmstadt, Germany.\\ %}
Email: saman.khodaverdian@rmr.tu-darmstadt.de}
% \received{v5.0 released November 2015}
}

\maketitle  

\begin{abstract}
The trajectory tracking problem for multivariable linear systems is considered. Two different techniques are examined: the \emph{output regulation theory} (ORT) and the \emph{flatness based tracking} (FBT). ORT and FBT are two different approaches to solve the tracking problem, and both methods have different restrictions. The tracking controller of the ORT furthermore depends on the solution of the so-called \emph{regulator equations}. In this paper, a special analytic solution of the regulator equations is presented. Additionally, based on this analytic solution, a link from the ORT to the FBT approach is provided, and the connection of both tracking controllers is highlighted. It is shown how the ORT controller can be converted to the FBT controller and that both methods lead to identical control laws for a certain class of systems. 
\end{abstract}


% ===============================================
\section{Introduction} \label{sec:introduction}
% ===============================================

This paper considers the tracking problem for multivariable linear systems. Tracking means that the system is controlled in a way such that its output follows a given reference trajectory. There is a considerable amount of literature dealing with this problem, where the tracking problem is usually solved by means of system inverses \citep{BrockettMesarovic1965, SainMassey1969, Silvermann1969, Hirschorn1979}. Based on these results, there have been developed different methods to solve the tracking problem. 

A proven way is the \emph{flatness based tracking} (FBT) approach \citep{Fliess1995, Fliess1999}, where a set of (virtual) outputs is searched such that the system states and inputs can be completely characterized by this so-called \emph{flat output} and its time derivatives. The tracking problem can easily be solved for such flat outputs. However, the conversion to the real system output might be problematic since for non-minimum phase systems the relation is determined by an unstable differential equation.\footnote
{ 
It is possible to overcome this problem, for instance, with approximate tracking techniques \citep{Hauser1992, Benvenuti1994} where non-minimum phase systems are approximated by minimum phase systems. However, this results in inexact tracking due to neglected terms. Another technique leading to exact tracking is the \emph{stable inversion theory} \citep{DevasiaChenPaden1996, ChenPaden1996}, but the solution needs non-causal and possibly unbounded control inputs \citep{Tomlin1995}. 
}
Nevertheless, the FBT is a common and widespread technique in control literature (see for instance \citep{Rothfuss1996, MartinDevasiaPaden1996, Fliess2000, Hagenmeyer2003, Hagenmeyer2004} and the references therein). 

Another approach to solve the tracking problem is provided by the \emph{output regulation theory} (ORT) which has been developed in \citep{FrancisWonham1976, Francis1977} for linear systems, and was extended to the nonlinear case in \citep{IsidoriByrnes1990}.\footnote{An alternative name for the ORT is the \emph{robust servomechanism problem} \citep{Davison1976}.} The ORT achieves stable tracking under mild assumptions (especially in the linear case), and bypasses the stability problems of non-minimum phase systems. The solution depends on the solvability of a set of equations, the \emph{regulator equations}, which is usually solved numerically. The drawback of the ORT is that the reference trajectory must be given by a so-called \emph{exosystem}, while in the FBT approach the reference trajectory can be any function that is sufficiently often differentiable. Thus, the ORT method is restricted in the treatable class of reference trajectories. Like the FBT, the ORT is a common technique for the design of 
tracking controllers (the textbooks \citep{SaberiStoorvogelSannuti2000} and \citep{Huang2004} are recommendable for a detailed overview). 

FBT and ORT are different approaches to solves the tracking problem, but in this paper it is shown that under certain conditions both methods lead to identical control laws. First, a special analytic solution of the regulator equations is presented which allows to interpret the tracking controller of the ORT in a new light. Based on this analytic solution, it is shown that ORT and FBT result in identical tracking controllers if the system output is flat, which also implies that in case of flat system outputs no exosystem is needed. Additionally, it is shown that in case of non-flat outputs, the ORT tracking controller can be converted to an FBT-like controller for a special choice of the free parameters. As expected, this controller parametrization leads to an unstable closed-loop behavior if the system is non-minimum phase. Note that the ORT approach has the ability to achieve stable tracking even in case of non-minimum phase systems. But then, the reference trajectory must be restricted to the class of exosystem 
trajectories, which can also be seen from the derived solution of the regulator equations.

In the following, some notations are summarized, before the problem setup and basics are stated in Section~\ref{sec:preliminaries}. The main result is presented in Section~\ref{sec:main}, where the single-input single-output (SISO) case is considered first (Section~\ref{sec:SISO}), before the general multiple-input multiple-output (MIMO) case is examined (Section~\ref{sec:MIMO}). Finally, the results are summarized in Section~\ref{sec:conclusion}.

{\bf Notation.}
Boldface letters denote vectors and matrices and italic letters represent scalar values. The identity matrix of dimension $n$ is written as $\I_n$, while $\zero$ denotes a zero matrix or vector of appropriate size. For a linear time-invariant system $\dot \x = \A \x + \B \u$, $\y = \C \x$, the invariant zeros are all numbers $\lambda$ for which the matrix
\begin{equation*}
 \R(\lambda) = \begin{bmatrix} \lambda\I_n-\A & -\B \\ \C & \zero \end{bmatrix} 
\end{equation*}
has a rank deficit \citep{MacFarlane1976}. $\R(\lambda)$ is known as the \emph{Rosenbrock's} system matrix. The system is called minimum phase if the real part of all invariant zeros is negative, otherwise it is called non-minimum phase.
The relative degree of the $k$-th output states how often $y_k$ must be differentiated with respect to time until the input $\u$ appears, i.e. $\delta_k = \min\limits_{\nu}\{\nu \ge 1: \C_k \A^{\nu-1} \B \ne \zero\}$ where $\C_k$ is the $k$-th row vector of $\C$. The relative degree of the system is then given by $\delta = \sum_{k=1}^m \delta_k$.

% ===============================================
\section{Problem statement and preliminaries} \label{sec:preliminaries}
% ===============================================

The trajectory tracking problem for multivariable linear systems is considered (as depicted in Figure~\ref{fig:setup}) where the plant is described by
\begin{subequations} \label{eq:plant}
 \begin{align}
  \dot\x &= \A\x + \B\u \\
  \y &= \C\x ,
 \end{align}
\end{subequations}
with state, input and output vector $\x \in \RZ^{n}$, $\u \in \RZ^{m}$ and $\y \in \RZ^{m}$, respectively. It is assumed that the plant is fully controllable and $\rank{\B} = \rank{\C} = m$. To avoid vastness, square systems are considered ($\dim{\u} = \dim{\y}$), but the results hold also for systems with more inputs than outputs. The goal is to find a feedback controller such that for any initial condition the output of the plant tracks a desired reference trajectory $\y^d$, i.e.
\begin{equation} \label{eq:goal}
 \lim\limits_{t \to \infty} (\y(t)-\y^d(t)) = \zero.
\end{equation}
The desired output vector $\y^d$ can be given by a trajectory generator as illustrated in Figure~\ref{fig:setup}. Two different approaches for solving the tracking problem \eqref{eq:goal} are summarized subsequently: the \emph{output regulation theory} (ORT) and the \emph{flatness based tracking} (FBT). 

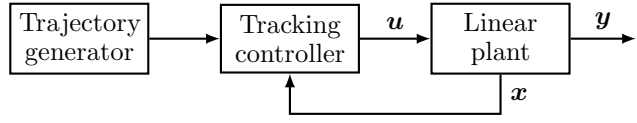
\begin{figure}[t]
 \centering
 \begin{tikzpicture}[scale=.25, auto]  
  %\scriptsize
  %\footnotesize
  \small
  
  \node[block, minimum height = 9mm, minimum width = 18mm, text width=16mm, align=center] (TG) {\linespread{1}\selectfont Trajectory generator\par};
  \node[block, right of=TG, node distance=28mm, minimum height = 9.0mm, minimum width = 18mm, text width=16mm, align=center] (C) {\linespread{1}\selectfont Tracking controller\par};
  \node[block, right of=C, node distance=28mm, minimum height = 9mm, minimum width = 18mm, text width=16mm, align=center] (P) {\linespread{1}\selectfont Linear plant\par};
  \node[coordinate, below of=P, node distance=10mm] (helpPoint) {};
  \node[coordinate, right of=P, node distance=18mm] (endPoint) {};

  % connections
  %-------------
  \draw[connector] (TG) -- (C);
  \draw[connector] (C) -- node {$\u$} (P);
  \draw[connector] (P) -- node {$\x$} (helpPoint) -| (C);
  \draw[connector] (P) -- node {$\y$} (endPoint);
  
 \end{tikzpicture}
 \caption{Basic setup for trajectory tracking.}
 \label{fig:setup}
\end{figure}

% ----------------------------------------
\subsection{Review of FBT} \label{sec:FBT}

A system is called flat if there exists a virtual output $\z$ such that all input and state variables can be expressed as functions of this output vector and its time derivatives, i.e. $\x = \f_x(\z, \dot{\z}, \ddot{\z}, \ldots)$ and $\u = \f_u(\z, \dot{\z}, \ddot{\z}, \ldots)$ \citep{Fliess1995, Fliess1999}. The virtual output $\z$ is called a flat output and does not necessarily need to coincide with the real output $\y$. 
% Thus, the flat output is possibly a virtual output. 
Some basics regarding the FBT of linear systems are summarized in the following. More details can be found in the mentioned references.  

\begin{proposition}
 A linear system \eqref{eq:plant} is flat if and only if it is controllable \citep{Fliess1995}.
\end{proposition}
Thus, for controllable linear systems there exists always a flat output which (in contrast to the nonlinear case) can be found very easily: since \eqref{eq:plant} is assumed to be controllable, there exists a linear transformation $\widetilde{\x} = \widetilde{\T}\x$ which brings the system into controllable canonical form \citep{Luenberger1967}, meaning that the transformed system matrices $(\widetilde{\A}, \widetilde{\B}, \widetilde{\C})=(\widetilde{\T}\A\widetilde{\T}^{-1}, \widetilde{\T}\B, \C\widetilde{\T}^{-1})$ are given by  
\begin{subequations} \label{eq:ccf}
\allowdisplaybreaks
\begin{alignat}{2}
 \widetilde{\A} &= \begin{bmatrix} \widetilde{\A}_{11} & \widetilde{\A}_{12} & \cdots & \widetilde{\A}_{1m} \\ \widetilde{\A}_{21} & \widetilde{\A}_{22} & \cdots & \widetilde{\A}_{2m} \\ \vdots & \vdots & \ddots & \vdots \\ \widetilde{\A}_{m1} & \widetilde{\A}_{m2} & \cdots & \widetilde{\A}_{mm} \end{bmatrix},&
\quad
 \widetilde{\B} = &\begin{bmatrix} \widetilde{\B}_{11} & \widetilde{\B}_{12} & \cdots & \widetilde{\B}_{1m} \\ \zero & \widetilde{\B}_{22} & \cdots & \widetilde{\B}_{2m} \\ \vdots & \vdots & \ddots & \vdots \\ \zero & \zero & \cdots & \widetilde{\B}_{mm} \end{bmatrix}, \\
 \widetilde{\A}_{ii} &= \begin{bmatrix} 0 & 1 &  \cdots & 0 \\ \vdots & \vdots & \ddots & \vdots \\ 0 & 0 & \cdots & 1 \\ \star & \star & \cdots & \star \end{bmatrix} \in \RZ^{\kappa_i \times \kappa_i},& 
\quad
 \widetilde{\A}_{\substack{ ij \\ i \ne j }} &= \begin{bmatrix} 0 & 0 &  \cdots & 0 \\ \vdots & \vdots & \ddots & \vdots \\ 0 & 0 & \cdots & 0 \\ \star & \star & \cdots & \star \end{bmatrix} \in \RZ^{\kappa_i \times \kappa_j}, \\ 
 \widetilde{\B}_{ii} &= \begin{bmatrix} 0 & \cdots & 0 & 1 \end{bmatrix}^\tp \in \RZ^{1 \times \kappa_i},& 
\quad
 \widetilde{\B}_{\substack{ ij \\ j>i }} &= \begin{bmatrix} 0 & \cdots & 0 & \star \end{bmatrix}^\tp \in \RZ^{1 \times \kappa_i}.
\end{alignat}
\end{subequations}
Herein, the $\kappa_i$'s are the controllablity (or \emph{Kronecker}) indices, with $\sum_{i=1}^m \kappa_i = n$, and $\star$ denotes elements that are not necessarily zero. The transformed output matrix $\widetilde{\C}$ does not have any specified structure. Obviously, a flat output is \begin{equation} \label{eq:z}
 \z = \begin{bmatrix} z_1 & z_2 & \cdots & z_m \end{bmatrix}^\tp = \begin{bmatrix} \widetilde{x}_1 & \widetilde{x}_{\kappa_1+1} & \cdots & \widetilde{x}_{n-\kappa_{m}+1} \end{bmatrix}^\tp
\end{equation}
since 
\begin{subequations} \label{eq:xz}
\begin{align}
 \begin{bmatrix} \widetilde{x}_1 & \widetilde{x}_2 & \cdots & \widetilde{x}_{\kappa_1} \end{bmatrix}^\tp &= \begin{bmatrix} z_1 & \dot z_1 & \cdots & z_1^{(\kappa_1-1)} \end{bmatrix}^\tp \\
 \begin{bmatrix} \widetilde{x}_{\kappa_1+1} & \widetilde{x}_{\kappa_1+2} & \cdots & \widetilde{x}_{\kappa_1+\kappa_2} \end{bmatrix}^\tp &= \begin{bmatrix} z_2 & \dot z_2 & \cdots & z_2^{(\kappa_2-1)} \end{bmatrix}^\tp \\
& \hspace{2mm} \vdots \notag \\ 
 \begin{bmatrix} \widetilde{x}_{n-\kappa_m+1} & \widetilde{x}_{n-\kappa_m+2} & \cdots & \widetilde{x}_n \end{bmatrix}^\tp &= \begin{bmatrix} z_m & \dot z_m & \cdots & z_m^{(\kappa_m-1)} \end{bmatrix}^\tp ,
\end{align}
\end{subequations}
with $z_i^{(k)} = {\text{d}^k z_i}/{\text{d}t^k}$. By a further derivation, and with the relations \eqref{eq:xz}, it can simply be seen that 
\begin{equation} \label{eq:uz}
 \u = \f_u(z_1, \dot{z}_1, \ldots, z_1^{(\kappa_1)}, \ldots, z_m, \dot{z}_m, \ldots, z_m^{(\kappa_m)}) . 
\end{equation}

Note that for the flat output $\z$, we have $\kappa_i = \delta_i$ which follows from \eqref{eq:ccf} and \eqref{eq:z}. Thus, $\delta = \sum_{i=1}^m \delta_i = \sum_{i=1}^m \kappa_i = n$ which implies that with the flat output there is no zero dynamics (since $\delta=n$). Given a sufficiently often differentiable reference trajectory $\z^d$, a tracking controller for the flat output can be designed such that $\lim\limits_{t \to \infty}(\z(t)-\z^d(t))=\zero$. Considering the error differential equations 
\begin{subequations} \label{eq:ez}
\begin{align} 
 (z_1^{(\kappa_1)} - z_1^{d \, (\kappa_1)}) + \sum_{k=0}^{\kappa_1-1} p_{1k} (z_1^{(k)} - z_1^{d \, (k)}) &= 0 \\ 
  & \hspace{2mm} \vdots \notag \\ 
 (z_m^{(\kappa_m)} - z_m^{d \, (\kappa_m)}) + \sum_{k=0}^{\kappa_m-1} p_{mk} (z_m^{(k)} - z_m^{d \, (k)}) &= 0
\end{align} 
\end{subequations}
and the relations \eqref{eq:xz} and \eqref{eq:uz}, a state feedback controller of the form 
\begin{equation} \label{eq:uFBT}
 \u = -\K \widetilde{\x} + \F \z^d_\Delta, 
\end{equation}
with $\z^d_\Delta = \begin{bmatrix} z^d_1 & \dot{z}^d_1 & \cdots & z_1^{d \, (\kappa_1)} & \cdots & z^d_m & \dot{z}^d_m & \cdots & z_m^{d \, (\kappa_m)} \end{bmatrix}^\tp$, can be found which regulates the flat output $\z$ to its desired trajectory $\z^d$. Herein, the controller parameters $p_{ik}$ must be chosen such that the dynamics \eqref{eq:ez} is stable.

\begin{remark}
 Flat outputs are not unique since every virtual output for which $\delta=n$ holds is a flat output \citep{MartinDevasiaPaden1996}.
\end{remark}

With the described approach, we get a tracking controller for $\z$, but the goal is to design a tracking controller for the real output $\y$. Note that the relationship from the virtual flat output to the real system output is given by
\begin{equation} \label{eq:yz}
 \y = \widetilde{\C} \begin{bmatrix} z_1 & \dot z_1 & \cdots & z_1^{(\kappa_1-1)} & \cdots & z_m & \dot z_m & \cdots & z_m^{(\kappa_m-1)} \end{bmatrix}^\tp .
\end{equation}
The differential equations \eqref{eq:yz} can be used to calculate the virtual reference output $\z^d$ from a desired reference output $\y^d$. Then, the feedback controller \eqref{eq:uFBT} regulates the flat output $\z$ to $\z^d$, which implies that $\y$ will be regulated to $\y^d$.

The problem of the FBT approach is that \eqref{eq:yz} describes the zero dynamics which is unstable in case of non-minimum phase systems. Thus, the solution of \eqref{eq:yz} might lead to unbounded reference signals $\z^d$ and, therefore, needs unbounded control inputs $\u$. 
A common way to bypass this problem is to approximate the virtual reference trajectory by a polynomial which fits the starting and ending point of the desired trajectory \citep{Fliess1998}. However, this method is only suitable for the change of the operating point along a specified trajectory, but not for the tracking of, e.g., sinusoidal signals. In addition, this approximation results in inexact tracking. 

\begin{remark}
 An alternative way to handle non-minimum phase systems is provided by the \emph{stable inversion theory} \citep{DevasiaChenPaden1996, ChenPaden1996}, but the obtained control input is non-causal and can lead to large input signals \citep{Tomlin1995}.
\end{remark}

For minimum phase systems, \eqref{eq:yz} describes stable zero dynamics and the calculation of $\z^d$ is unproblematic. In this case, the FBT approach provides an intuitive way to design a trajectory tracking controller.

% ----------------------------------------
\subsection{Review of ORT} \label{sec:ORT}

For the ORT approach, it is assumed that the trajectory to be tracked is described by an autonomous linear system (the exosystem)  
\begin{subequations} \label{eq:exo}
 \begin{align}
  \dot\o &= \S\o \\
  \y^d &= \Q\o ,
 \end{align}
\end{subequations}
with $\o \in \RZ^{r}$, $\y^d \in \RZ^{m}$. Like in the FBT approach, the goal is to design a feedback controller 
\begin{equation} \label{eq:uORT_KF}
 \u = -\K\x + \F\o
\end{equation}
which leads to a stable closed-loop system and achieves output regulation in the sense of \eqref{eq:goal}. Without loss of generality, it is assumed that all eigenvalues of $\S$ have nonnegative real part since eigenmodes with negative real part vanish asymptotically. According to \citep{FrancisWonham1976, Francis1977}, the following proposition is formulated. 

\begin{proposition} \label{prop:ORT}
 The output regulation problem is solvable if and only if there exists a matrix pair $(\mPi,\mGamma)$ solving the so-called \emph{regulator equations} (sometimes called \emph{Francis equations})
 \begin{subequations} \label{eq:reg}
 \begin{align} 
  \mPi\S &= \A\mPi + \B\mGamma \label{eq:reg1} \\
  \Q &= \C\mPi. \label{eq:reg2}
 \end{align}
 \end{subequations}
 Then, control law \eqref{eq:uORT_KF} achieves output regulation (or tracking) if $\K$ is chosen such that $\A-\B\K$ is stable and $\F=\K\mPi+\mGamma$.
\end{proposition}
Since Proposition~\ref{prop:ORT} is a well known result, the proof is omitted and can be found for instance in \citep{SaberiStoorvogelSannuti2000}.

\begin{remark}
 The feedback matrix $\K$ in \eqref{eq:uORT_KF} must be determined such that $\A-\B\K$ is stable. Thus, for the ORT approach it is sufficient to have a stabilizable system instead of a controllable system.  
\end{remark}

The next proposition gives information about the solvability of the matrix equations \eqref{eq:reg}.
\begin{proposition} \label{prop:reg_eq}
 The regulator equations \eqref{eq:reg} are solvable if and only if 
 \begin{equation*}
  \rank{\R(\lambda)} = \rank{\begin{bmatrix} \lambda\I_n-\A & -\B \\ \C & \zero \end{bmatrix}} = n+m 
 \end{equation*}
 for every $\lambda$ which is an eigenvalue of $\S$. 
\end{proposition}
This follows from a result presented in \citep{Hautus1983}. In this paper, an alternative proof will be derived in Section~\ref{sec:main}. Note that $\R(\lambda)$ is the \emph{Rosenbrock's} system matrix. Therefore, the rank condition in Proposition~\ref{prop:reg_eq} implies that  
none of the systems invariant zeros coincide with an eigenvalue of $\S$. 

Summing up, we conclude that if the reference trajectory is described by an exosystem of the form \eqref{eq:exo}, tracking can be achieved by a feedback controller \eqref{eq:uORT_KF} if and only if the rank condition in Proposition~\ref{prop:reg_eq} holds.

% \begin{remark}
It should be noted that the classical ORT can be used not only to track a desired reference trajectory, but also to reject disturbances. This will be achieved if the exosystem is extended by the disturbance dynamics that need to be rejected. However, since we are interested in output tracking, the disturbance rejection problem is not considered in the following.
% \end{remark}

The problem of the ORT approach is that $\y^d$ must be determined by an exosystem of the form \eqref{eq:exo}. This restricts the class of treatable reference trajectories. For the FBT approach, the reference trajectory can be any function which is sufficiently often differentiable, however, with the drawback that -- in contrast to the ORT approach -- it is not directly applicable to non-minimum phase systems. 

\begin{remark}
 The important class of \emph{Bohl functions} (combinations of polynomial, exponential and sinusoidal signals) can be represented by exosystems \citep{Trentelman2001}. %(p.195)
\end{remark}

% ===============================================
\section{From ORT to FBT} \label{sec:main}
% ===============================================

In the following, an analytic solution of the regulator equations \eqref{eq:reg} is presented. Based on this analytic solution, it is shown that the tracking controller of the ORT is identical to the tracking controller of the FBT for systems with $\delta = n$. Furthermore, the difficulties in case of $\delta < n$ are highlighted, and it is shown how the controller parameters given from ORT and FBT are correlated. 
For reasons of clarity, SISO systems are considered at first, before the general MIMO case is examined.

% ----------------------------------------
\subsection{SISO systems} \label{sec:SISO}

Let system \eqref{eq:plant} be a SISO system with transfer function 
\begin{equation} \label{eq:tf}
 g(s) = \frac{b_0 + b_1 \cdot s + \ldots + b_\tau \cdot s^\tau}{a_0 + a_1 \cdot s + \ldots + a_{n-1} \cdot s^{n-1} + s^n}, \quad \tau < n . 
\end{equation}
Without loss of generality, it is assumed that the system matrices are given in controllable canonical form, i.e.
\begin{equation} \label{eq:ccfSISO}
 \A = \begin{bmatrix} 0 & 1 & \cdots & 0 \\ \vdots & \vdots & \ddots & \vdots \\ 0 & 0 & \cdots & 1 \\ -a_0 & -a_1 & \cdots & -a_{n-1} \end{bmatrix}, \quad \B = \begin{bmatrix} 0 \\ \vdots \\ 0 \\ 1 \end{bmatrix}, \quad \C^\tp = \begin{bmatrix} b_0 \\ \vdots \\ b_\tau \\ \zero \end{bmatrix} .
\end{equation}
Then, the regulator equations \eqref{eq:reg} read
\begin{subequations} \label{eq:regSISO}
\allowdisplaybreaks
\begin{align}
 \begin{bmatrix} \vpi_1 \\ \vpi_{2} \\ \vdots \\ \vpi_n \end{bmatrix} \S &= \begin{bmatrix} 0 & 1 & \cdots & 0 \\ \vdots & \vdots & \ddots & \vdots \\ 0 & 0 & \cdots & 1 \\ -a_0 & -a_1 & \cdots & -a_{n-1} \end{bmatrix} \begin{bmatrix} \vpi_1 \\ \vpi_{2} \\ \vdots \\ \vpi_n \end{bmatrix} + \begin{bmatrix} 0 \\ \vdots \\ 0 \\ 1 \end{bmatrix} \mGamma  \label{eq:regSISO1} \\
 \Q &= \begin{bmatrix} b_0 & \cdots & b_\tau & 0 & \cdots & 0 \end{bmatrix} \begin{bmatrix} \vpi_1 \\ \vpi_{2} \\ \vdots \\ \vpi_n \end{bmatrix} , \label{eq:regSISO2} 
\end{align}
\end{subequations}
where the matrix $\mPi$ is split into its row vectors. From the first $n-1$ equations in \eqref{eq:regSISO1}, it follows that
\begin{equation} \label{eq:pi_k}
 \vpi_k = \vpi_1 \S^{k-1}, \quad k = 2,\ldots, n .
\end{equation}
Taking this into account, the last equation in \eqref{eq:regSISO1} can be written as
\begin{equation}
  \mGamma = \vpi_1 \underbrace{(a_0\I_r + a_1\S + \ldots + a_{n-1}\S^{n-1} + \S^n)}_{\N} = \vpi_1 \N .
\end{equation}
Furthermore, from \eqref{eq:regSISO2} and \eqref{eq:pi_k}, we get
\begin{equation} \label{eq:Q-pi_1}
 \Q = \vpi_1 \underbrace{(b_0\I_r + b_1\S + \ldots + b_{\tau}\S^{\tau})}_{\M} = \vpi_1 \M .
\end{equation}
Note that the coefficients $b_0,\ldots,b_\tau$ are the coefficients of the zero polynomial of the plant (cf. \eqref{eq:tf}). From the \emph{Cayley-Hamilton} theorem \citep{Gantmacher1959}, it is clear that $\M$ is regular if and only if no plant zero coincides with an eigenvalue of the exosystem. In this case, \eqref{eq:Q-pi_1} yields 
\begin{equation}
 \vpi_1 = \Q \M^{-1} 
\end{equation}
and, finally, the solution of the regulator equations reads 
\begin{equation} \label{eq:regSISOsol}
 \mPi = \begin{bmatrix} \Q \M^{-1} \\ \Q \M^{-1} \S \\ \vdots \\ \Q \M^{-1}\S^{n-1} \end{bmatrix}, \quad \mGamma = \Q \M^{-1} \N . 
\end{equation}
The above derivation is also a proof for Proposition~\ref{prop:reg_eq}.

\begin{theorem} \label{thm:ORT_FBT_SISO}
 Given system \eqref{eq:tf} or \eqref{eq:ccfSISO} has no zero dynamics, i.e. $\delta = n$, ORT and FBT lead to the same control law. 
\end{theorem}

\begin{proof}
 First note that $\tau = 0$ for $\delta = n$,  
 and $y = b_0 \cdot x_1$ is a flat output. Thus, as described for the FBT approach, taking the error dynamics 
 \begin{equation} \label{eq:error_dyn}
  (y^{(n)} - y^{d \, (n)}) + \ldots + p_1(\dot{y} - \dot{y}^d) + p_0(y - y^d) = 0
 \end{equation}
 and considering \eqref{eq:ccfSISO}, we get 
 \begin{align} \label{eq:uFBT_SISO}
  u &= -k_1 \cdot x_1 - k_2 \cdot x_2 -\ldots - k_n \cdot x_n+\frac{1}{b_0}(p_0 \cdot y^d + p_1 \cdot \dot{y}^d+\ldots+y^{d \, (n)}) \notag \\
    &= -\underbrace{\begin{bmatrix} k_1 & k_2 & \cdots & k_n \end{bmatrix}}_{\K} \x + \underbrace{\frac{1}{b_0} \begin{bmatrix} p_0 & p_1 & \cdots & 1 \end{bmatrix}}_{\F} \y^d_\Delta ,
 \end{align}
 with $k_i = p_{i-1}-a_{i-1}$ and $\y^d_\Delta = \begin{bmatrix} y^d & \dot{y}^d & \cdots & y^{d \, (n)} \end{bmatrix}^\tp$. The controller gain $\K$ and pre-filter $\F$ are determined by the parameters $p_k$ which are the coefficients of the specified error dynamics \eqref{eq:error_dyn}.  

 In the ORT approach, the tracking controller is given by
 \begin{equation} \label{eq:uORT}
  u = -\K \x + (\K\mPi+\mGamma) \o.  
 \end{equation}
 Furthermore, with \eqref{eq:regSISOsol} and $\M^{-1} = \frac{1}{b_0}\I_r$, we have 
 \begin{subequations} \label{eq:PoGp}
 \begin{align}
  \mPi\o &= \frac{1}{b_0} \begin{bmatrix} \Q\o \\ \Q\S\o \\ \vdots \\ \Q\S^{n-1}\o \end{bmatrix} = \frac{1}{b_0} \begin{bmatrix} y^d \\ \dot{y}^d \\ \vdots \\ y^{d \, (n-1)} \end{bmatrix} \\
  \mGamma\o &= \frac{1}{b_0} \Q (a_0\I_r + a_1\S + \ldots + a_{n-1}\S^{n-1} + \S^n) \o \\ 
  &= \frac{1}{b_0} (a_0 \cdot y^d + a_1 \cdot \dot{y}^d + \ldots + y^{d \, (n)}) .
 \end{align}
 \end{subequations}
 Therefore, the tracking controller can be written as 
 \begin{equation} \label{eq:uORT_SISO}
  u = - \underbrace{\begin{bmatrix} k_1 & k_2 & \cdots & k_n \end{bmatrix}}_{\K} \x + \underbrace{\frac{1}{b_0}\begin{bmatrix} k_1+a_0 & k_2+a_1 & \cdots & 1 \end{bmatrix}}_{\F} \y^d_\Delta .
 \end{equation}
 Substituting $k_i + a_{i-1} = p_{i-1}$, \eqref{eq:uORT_SISO} is exactly the same as \eqref{eq:uFBT_SISO}.  
\end{proof}

% \begin{remark}
It should be noted that the coefficients of the closed-loop characteristic polynomial are $k_i + a_{i-1}$ (since the system is in controllable canonical form). Thus, choosing $k_i = p_{i-1} - a_{i-1}$, the closed-loop characteristic polynomial and the tracking error dynamics \eqref{eq:error_dyn} are both determined by the coefficients $p_{i-1}$. 
% \end{remark}

The above result shows that for systems without zero dynamics, ORT and FBT lead to the same control law. Thus, the following is a logical consequence.

\begin{corollary} \label{col:exo}
 In the ORT approach, the reference trajectory does not need to be determined by an exosystem of the form \eqref{eq:exo} if the plant has no zero dynamics.
\end{corollary}

For systems without zero dynamics, we only need to know $y^d$ and its time derivatives up to the order $y^{d \, (n)}$, like in the FBT approach. If there are zero dynamics, i.e. $\tau>0$, reformulating the feedforward part of the tracking controller as in \eqref{eq:PoGp} is not possible and the exosystem matrices $\S$ and $\Q$ appear in the control law \eqref{eq:uORT_SISO}. 
% In this case, the reference trajectory is determined by the exosystem matrices. 
The next result presents how the ORT controller can be converted into a FBT-like control law even if the plant is determined by zero dynamics. 

\begin{corollary} \label{col:cancel_pz}
 An exosystem formulation can be omitted if the plant zeros are canceled by closed-loop poles.
\end{corollary}

\begin{proof}
 Given that $\tau > 0$ and taking \eqref{eq:regSISOsol} into account, the tracking controller \eqref{eq:uORT} is written as
 \begin{equation} \label{eq:uORT_SISO_gen}
  u = -\K\x + \Q\M^{-1} \underbrace{\Big( \widetilde{k}_1 \I_r + \ldots + \widetilde{k}_n \S^{n-1} + \S^n \Big)}_{\K_S} \o,
 \end{equation}
 with $\widetilde{k}_i = k_i+a_{i-1}$. Furthermore, the controller gains can be chosen such that 
 \begin{equation} \label{eq:KS}
  \K_S = \Big( \overline{k}_1 \I_r + \ldots + \overline{k}_{\tau+1} \S^{\tau} \Big) \Big( \widehat{k}_1 \I_r + \ldots + \widehat{k}_{\delta} \S^{\delta-1} + \frac{1}{\overline{k}_{\tau+1}} \S^{\delta} \Big),
 \end{equation}
 where $\overline{k}_i$ and $\widehat{k}_j$ are the free parameters and $\delta = n-\tau$. Including \eqref{eq:KS} into \eqref{eq:uORT_SISO_gen}, we see that for $\overline{k}_i = b_{i-1}$ the matrix $\M^{-1}$ cancels out, leading to
 \begin{subequations} \label{eq:uORT_SISO_red}
 \begin{align} 
  u &= -\K\x + \Q \Big( \widehat{k}_1 \I_r + \ldots + \widehat{k}_{\delta} \S^{\delta-1} + \frac{1}{b_\tau} \S^{\delta} \Big) \o \\
    &= -\K\x +\Big( \widehat{k}_1 \cdot y^d +\ldots +\widehat{k}_{\delta} \cdot y^{d \, (\delta-1)} +\frac{1}{b_\tau} \cdot y^{d \, (\delta)} \Big).
 \end{align}
 \end{subequations}
 Note that the choice $\overline{k}_i = b_{i-1}$ implies pole-zero cancellations in the closed-loop system $(\A-\B\K,\B,\C)$. However, \eqref{eq:uORT_SISO_red} depends only on $y^d, \dot{y}^d, \ldots, y^{d \, (\delta)}$ and the exosystem can be omitted.  
\end{proof}

Corollary~\ref{col:cancel_pz} is not surprising since the transfer function of the closed-loop system is given by 
\begin{equation*}
 g(s) =\frac{b_0 + \ldots + b_\tau \cdot s^\tau}{\widetilde{k}_1 + \widetilde{k}_2 \cdot s + \ldots + s^n} 
      =\frac{b_0 + \ldots + b_\tau \cdot s^\tau}{\overline{k}_1 + \ldots + \overline{k}_{\tau+1} \cdot s^\tau} \cdot \frac{1}{\widehat{k}_1 + \ldots + \frac{1}{\overline{k}_{\tau+1}} s^{\delta}} ,
\end{equation*}
 and choosing $\overline{k}_i = b_{i-1}$ yields
\begin{equation} \label{eq:tf_cl}
 g(s) = \frac{1}{\widehat{k}_1 + \ldots + \widehat{k}_{\delta} s^{\delta-1} + \frac{1}{b_\tau} s^{\delta}} .
\end{equation}
Transfer function \eqref{eq:tf_cl} describes a zero free system with reduced dynamical order $\delta=n-\tau$ for which Theorem~\ref{thm:ORT_FBT_SISO} is applicable. 

It is clear that Corollary~\ref{col:cancel_pz} is only useful for minimum phase systems since for non-minimum phase systems the cancellation of zeros leads to an unstable behavior. Summarizing the results of Corollary~\ref{col:exo} and \ref{col:cancel_pz}, we conclude that the ORT tracking controller can be converted into a FBT-like control law. This also implies that for the ORT approach, the exosystem can be omitted in case of minimum phase systems. Only $y^d$ and its time derivatives up to order $y^{d \, (\delta)}$ are needed for output tracking as in the FBT approach. 
It should be noted that for the ORT approach, the controller gain does not need to be determined as in Corollary~\ref{col:cancel_pz}, but then the reference trajectory must be provided by an exosystem. Thus, in contrast to the FBT approach, the ORT approach can handle unstable zero dynamics, albeit at the expense of the treatable class of reference trajectories.

The results presented for SISO systems are extended to the MIMO case in the following. %, but the derivations are more demanding. 

% ----------------------------------------
\subsection{MIMO systems} \label{sec:MIMO}

Since system \eqref{eq:plant} is controllable, we can transform its system matrices into the MIMO controllable canonical form \eqref{eq:ccf}. In addition, it is always possible to choose $\u = -\widetilde{\K} \widetilde{\x} + \widetilde{\F} \widetilde{\u}$ such that
\begin{equation} \label{eq:A'B'}
 \dot{\widetilde{\x}} = \underbrace{( \widetilde{\A}-\widetilde{\B}\widetilde{\K} )}_{\A'} \widetilde{\x} + \underbrace{\widetilde{\B}\widetilde{\F}}_{\B'} \widetilde{\u} %\\
\end{equation}
with 
\begin{subequations} \label{eq:bcf}
\allowdisplaybreaks
\begin{align}
 \A' = \begin{bmatrix} \A'_{11} & \zero & \cdots & \zero \\ \zero & \A'_{22} & \cdots & \zero \\ \vdots & \vdots & \ddots & \vdots \\ \zero & \zero & \cdots & \A'_{mm} \end{bmatrix}, \quad
 \B' = &\begin{bmatrix} \B'_{11} & \zero & \cdots & \zero \\ \zero & \B'_{22} & \cdots & \zero \\ \vdots & \vdots & \ddots & \vdots \\ \zero & \zero & \cdots & \B'_{mm} \end{bmatrix} , \label{eq:bcf_sys} \\
 \A'_{ii} = \begin{bmatrix} 0 & 1 &  \cdots & 0 \\ \vdots & \vdots & \ddots & \vdots \\ 0 & 0 & \cdots & 1 \\ 0 & 0 & \cdots & 0 \end{bmatrix}, 
 \qquad
 &\B'_{ii} = \begin{bmatrix} 0 \\ \vdots \\ 0 \\ 1 \end{bmatrix}. \label{eq:bnf_ii}
\end{align}
\end{subequations}
This means that system \eqref{eq:A'B'} is described by integrator chains of length $\kappa_i$. 
Note that $\widetilde{\C}$ has no specific structure. Since such input and state transformations are always possible for controllable systems, in the following it is assumed that the MIMO system \eqref{eq:plant} is determined by
\begin{subequations} \label{eq:ccfMIMO}
\begin{align}
 \dot{\x} &= \begin{bmatrix} \A_{11} & \cdots & \zero \\ \vdots & \ddots & \vdots \\ \zero & \cdots & \A_{mm} \end{bmatrix} \x + \begin{bmatrix} \B_{11} & \cdots & \zero \\ \vdots & \ddots & \vdots \\ \zero & \cdots & \B _{mm} \end{bmatrix} \u \\
 \y &= \begin{bmatrix} \C_{11} & \cdots & \C_{1m} \\ \vdots & \ddots & \vdots \\ \C_{m1} & \cdots & \C_{mm} \end{bmatrix} \x ,
\end{align}
\end{subequations}
where $\A_{ii}$ and $\B_{ii}$ are given as in \eqref{eq:bcf}. (The symbols $\; _{\widetilde{\phantom{aa}}} \;$ and $\; \prime \;$ are omitted for simplicity.)

Furthermore, we assume that the exosystem matrices are specified by 
\begin{equation} \label{eq:exoMIMO}
 \S = \begin{bmatrix} \S_{1} & \cdots & \zero \\ \vdots & \ddots & \vdots \\ \zero & \cdots & \S_{m} \end{bmatrix} 
  \quad \text{ and } \quad  
 \Q = \begin{bmatrix} \Q_{1} & \cdots & \zero \\ \vdots & \ddots & \vdots \\ \zero & \cdots & \Q_{m} \end{bmatrix}, 
\end{equation}
i.e. every reference output is described by its own exosystem $\S_k \in \RZ^{r_k \times r_k}$, $\Q_k \in \RZ^{1 \times r_k}$. This assumption is not restrictive since every output has (usually) its own reference trajectory.\footnote{The assumption can be relaxed for the subsequent steps, but it considerably simplifies the derivations.} 

Partitioning the matrices $\mPi$ and $\mGamma$ in \eqref{eq:reg} as
\begin{equation} \label{eq:PiGamma}
 \begin{bmatrix} \mPi_{11} & \cdots & \mPi_{1m} \\ \vdots & \ddots & \vdots \\ \mPi_{m1} & \cdots & \mPi_{mm} \end{bmatrix} 
 \quad \text{ and } \quad
 \begin{bmatrix}\mGamma_{11} &\cdots & \mGamma_{1m} \\ \vdots & \ddots & \vdots \\ \mGamma_{m1} & \cdots & \mGamma_{mm}\end{bmatrix},
\end{equation}
it is not hard to see that 
\begin{equation}
 \mPi_{ij} \S_j = \A_{ii} \mPi_{ij} + \B_{ii} \mGamma_{ij}, \quad i,j \in \{1,\ldots,m\}. 
\end{equation}
Note that $\A_{ii}$ and $\B_{ii}$ are given as in \eqref{eq:bnf_ii}. With the result presented in Section~\ref{sec:SISO}, it follows immediately that 
\begin{equation} \label{eq:PiGamma_sol}
 \mPi_{ij} = \begin{bmatrix} \vpi_{ij}^1 \\ \vpi_{ij}^1 \S_j \\ \vdots \\ \vpi_{ij}^1 \S_j^{\kappa_i-1} \end{bmatrix} 
 \quad \text{ and } \quad 
 \mGamma_{ij} = \vpi_{ij}^1 \S_j^{\kappa_i} . 
\end{equation}
Hence, the only unknowns are the row vectors $\vpi_{ij}^1$ ($i,j \in \{1,\ldots,m\}$), which can be determined from the second regulator equation \eqref{eq:reg2}:
\begin{equation} \label{eq:reg2MIMO}
 \begin{bmatrix} \Q_{1} & \cdots & \zero \\ \vdots & \ddots & \vdots \\ \zero & \cdots & \Q_{m} \end{bmatrix} 
 =  \begin{bmatrix} \sum_{l=1}^m \C_{1l}\mPi_{l1} & \cdots & \sum_{l=1}^m \C_{1l}\mPi_{lm} \\ \vdots & \ddots & \vdots \\ \sum_{l=1}^m \C_{ml}\mPi_{l1} & \cdots & \sum_{l=1}^m \C_{ml}\mPi_{lm} \end{bmatrix} .
\end{equation}

Let $\C_{ij} = \begin{bmatrix} c_{ij}^{0} & c_{ij}^{1} & \cdots & c_{ij}^{\kappa_j-1} \end{bmatrix}$, the entries in the $k$-th column block of \eqref{eq:reg2MIMO} are given by
\begin{equation} \label{eq:kil}
 \sum_{l=1}^m \vpi_{lk}^1 \underbrace{\sum_{\nu=0}^{\kappa_l-1} c_{il}^\nu \S_k^\nu}_{\M_{kil}} = \begin{cases}\Q_k, &i=k, \\ \zero, &i \ne k.\end{cases}
\end{equation}
The index in $\M_{kil} \in \RZ^{r_k \times r_k}$ denotes the $l$-th summand of the $i$-th block entry in the $k$-th column block. Now, \eqref{eq:kil} can be rearranged to
\begin{subequations}
\begin{align}
 \begin{bmatrix}\vpi_{1k}^1 &\cdots& \vpi_{mk}^1 \end{bmatrix} \M_k = \begin{bmatrix}\zero &\cdots& \Q_k &\cdots& \zero \end{bmatrix} \\
 \intertext{with the square matrix}
 \M_k = \begin{bmatrix} \M_{k11} & \cdots & \M_{km1} \\ \vdots & \ddots & \vdots \\ \M_{k1m} & \cdots & \M_{kmm} \end{bmatrix} \in \RZ^{m r_k \times m r_k} .   
\end{align}
\end{subequations}
Thus, given that $\M_k$ is regular, the solution reads
\begin{equation} \label{eq:pi1k_sol}
 \begin{bmatrix}\vpi_{1k}^1 &\cdots& \vpi_{mk}^1\end{bmatrix} = \begin{bmatrix}\zero &\cdots& \Q_k &\cdots& \zero\end{bmatrix} \M_k^{-1}. 
\end{equation}

\begin{remark}
 A simple construction rule for $\M_k$ is
 \begin{subequations} \label{eq:Mk}
 \begin{align}
  \M_k &= \overline{\S}_{k} (\C^\tp \otimes \I_{r_k}) %\overline{C}^\tp_k 
  \intertext{where $\overline{\S}_{k} \in \RZ^{m r_k \times n r_k}$ is defined as}
  \overline{\S}_{k} &= \begin{bmatrix} \I_{r_k} & \S_k & \cdots & \S_k^{\kappa_1-1} & \cdots & \zero & \zero & \cdots & \zero \\ \vdots & \vdots &  & \vdots & \ddots & \vdots & \vdots &  & \vdots \\ \zero & \zero & \cdots & \zero & \cdots & \I_{r_k} & \S_k & \cdots & \S_k^{\kappa_m-1} \end{bmatrix} .
 \end{align}
 \end{subequations}
\end{remark}

From \eqref{eq:pi1k_sol}, it is clear that the solvability of the regulator equations depends on the invertibility of the matrices $\M_k$, $k = 1,\ldots,m$. Hence, we have to check the regularity of $\M_k$. 

\begin{proposition} \label{prop:Mk}
 $\M_k$ is regular if and only if no invariant system zero coincides with an exosystem eigenvalue, i.e. if $\rank{\R(\lambda)}=n+m$ for every $\lambda$ which is an eigenvalue of $\S$.
\end{proposition}

The proof of Proposition~\ref{prop:Mk} is given in Appendix~\ref{app:prop:Mk}. 
The following theorem summarizes the results.

\begin{theorem} \label{thm:reg_eq_sol}
 Given a plant in canonical form \eqref{eq:bcf}--\eqref{eq:ccfMIMO} and let the exosystem be described as in \eqref{eq:exoMIMO}. Then, provided that $\M_k$ (defined as in \eqref{eq:Mk}) is regular, the solution of the regulator equations is determined by \eqref{eq:PiGamma_sol} and \eqref{eq:pi1k_sol}.
\end{theorem}

\begin{remark}
 Taking \eqref{eq:A'B'} and $(\widetilde{\A}, \widetilde{\B}, \widetilde{\C})=(\widetilde{\T}\A\widetilde{\T}^{-1}, \widetilde{\T}\B, \C\widetilde{\T}^{-1})$ into account, the solution of the regulator equations for the original system \eqref{eq:plant} is given by
 \begin{equation*}
  \mPi_\text{orig} = \widetilde{\T}^{-1} \mPi,  \qquad  \mGamma_\text{orig} = \widetilde{\F}\mGamma - \widetilde{\K}\mPi .
 \end{equation*}
\end{remark}

Theorem~\ref{thm:reg_eq_sol} provides an analytic solution of the regulator equations for MIMO systems in controllable canonical form. However, the canonical form \eqref{eq:ccfMIMO} is not well suited to compare the ORT and FBT methods.
Therefore, given the system output is flat, an alternative solution will be presented which is better suited to compare both methods. 
% Given that $\y$ is flat, we consider the transformation matrix
% For this reason, an alternative solution will be presented which is better suited in case of flat system outputs. 
Let us consider the transformation matrix 
\begin{equation} \label{eq:T}
 \T = \begin{bmatrix} \C_1 \\ \vdots \\ \C_1 \A^{\delta_1-1} \\ \vdots \\ \C_m \\ \vdots \\ \C_m \A^{\delta_m-1} \end{bmatrix} \in \RZ^{\delta \times n} ,
\end{equation}
where $\C_k$ is the $k$-th row vector of $\C$. Note that $\rank{\T}=\delta$ (cf. Appendix~\ref{app:T}).
Furthermore, assume that $\y$ is a flat output such that $\delta=n$ and $\T$ is a square matrix.
Applying the transformation \eqref{eq:T}, we get $\A' = \T\A\T^{-1}$ as in \eqref{eq:bcf} and
\begin{equation}
 \C' = \C\T^{-1} = \begin{bmatrix} \C'_{11} & \cdots & \zero \\ \vdots & \ddots & \vdots \\ \zero & \cdots & \C'_{mm} \end{bmatrix}, \quad \text{with} \quad \C'_{kk} = \begin{bmatrix} 1 & 0 & \cdots & 0 \end{bmatrix}.
\end{equation}
Due to the definition of the relative degrees, the input matrix changes to
\begin{equation} \label{eq:B'}
 \B' = \T\B = \begin{bmatrix} \zero & \cdots & (\d_1^\star)^\tp & \cdots & \zero & \cdots & (\d_m^\star)^\tp \end{bmatrix}^\tp, 
\end{equation}
with $\d_k^\star = \C_k \A^{\delta_k-1} \B$. It can easily be checked that $\A'$ has the same integrator chain form as before. 
The transformed system $(\A',\B',\C')$ is suitable to prove the following claim.

\begin{theorem} \label{thm:ORT_FBT_MIMO}
 ORT and FBT approach lead to the same control law for systems with $\delta=n$, i.e. Theroem~\ref{thm:ORT_FBT_SISO} also holds for MIMO systems.
\end{theorem}

\begin{proof}
 First note that the output of $(\A',\B',\C')$ is a flat output (since $\delta=n$). Thus, as described in Section~\ref{sec:FBT}, for the FBT controller we consider the error differential equations
 \begin{subequations} \label{eq:ey}
 \begin{align} 
  (y_1^{(\delta_1)} - y_1^{d \, (\delta_1)}) + \sum_{k=0}^{\delta_1-1} p_{1k} (y_1^{(k)} - y_1^{d \, (k)}) &= 0 \\ 
   & \hspace{2mm} \vdots \notag \\ 
  (y_m^{(\delta_m)} - y_m^{d \, (\delta_m)}) + \sum_{k=0}^{\delta_m-1} p_{mk} (y_m^{(k)} - y_m^{d \, (k)}) &= 0 ,
 \end{align} 
 \end{subequations}
 leading to
 \begin{equation} \label{eq:ey_compact}
   \D^\star \u - \y^d_\delta + \P(\x-\y^d_{\nabla}) = \zero ,
 \end{equation} 
 with
 \begin{align*}
  \D^\star &= \begin{bmatrix} (\d_1^\star)^\tp & \cdots & (\d_m^\star)^\tp \end{bmatrix}^\tp,  \quad  \y^d_\delta = \begin{bmatrix} y_1^{d \, (\delta_1)} & \cdots & y_m^{d \, (\delta_m)} \end{bmatrix}^\tp \\
  \y^d_\nabla &= \begin{bmatrix} y^d_1 & \dot{y}^d_1 & \cdots & y_1^{d \, (\delta_1-1)} & \cdots & y^d_m & \dot{y}^d_m & \cdots & y_m^{d \, (\delta_m-1)} \end{bmatrix}^\tp , \\ 
  \P &= \begin{bmatrix} p_{10} & \cdots & p_{1 \delta_1-1} & \cdots & 0 & \cdots & 0 \\ \vdots & & \vdots & \ddots & \vdots & & \vdots \\ 0 & \cdots & 0 & \cdots & p_{m0} & \cdots & p_{m \delta_m-1} \end{bmatrix} .
 \end{align*}
 Therefore, the control law
 \begin{equation} \label{eq:uFBT_MIMO}
  \u = - (\D^\star)^{-1} \left( \P(\x-\y^d_\nabla) - \y^d_\delta \right) 
 \end{equation}
 achieves the desired tracking behavior. Note that the inverse $(\D^\star)^{-1}$ exists, which follows from \eqref{eq:B'}. Since $\T$ is a regular matrix and $\rank{\B}=m$, it is clear that $\rank{\B'}=m$ and, therefore, $\rank{\D^\star}=m$. 

 Next, we consider the ORT approach. Apparently, for $\u = (\D^\star)^{-1} \u'$ we get completely decoupled SISO systems where the $k$-th SISO system is zero free with dynamical order $\delta_k$ (the $k$-th transfer function reads $y_k(s)/u'_k(s) = 1/s^{\delta_k}$). Thus, we can use the results of Section~\ref{sec:SISO} and the ORT tracking controller for the $k$-th SISO system is given by (cf. \eqref{eq:uORT_SISO})
 \begin{equation}
  u'_k = -\begin{bmatrix} p_{k0} & \cdots & p_{k \delta_k-1} \end{bmatrix} \x_k + \begin{bmatrix} p_{k0} & \cdots & p_{k \delta_k-1} & 1 \end{bmatrix} \y^d_{\Delta,k}
 \end{equation} 
 or after some rearrangements 
 \begin{equation}
  u'_k = - \left( \begin{bmatrix} p_{k0} & \cdots & p_{k \delta_k-1} \end{bmatrix} (\x_k - \y^d_{\nabla,k}) - y_k^{d \, (\delta_k)} \right) ,
 \end{equation}
 with  
 \begin{align*}
  \y^d_{\Delta,k} &= \begin{bmatrix} (\y^d_{\nabla,k})^\tp & y_k^{d \, (\delta_k)} \end{bmatrix}^\tp , \\
  \y^d_{\nabla,k} &= \begin{bmatrix} y^d_k & \dot{y}^d_k & \cdots & y_k^{d \, (\delta_k-1)} \end{bmatrix}^\tp , \\
  \x_k &= \begin{bmatrix} x_{\sum_{j=1}^{k-1} \delta_j + 1} & \cdots & x_{\sum_{j=1}^{k} \delta_j}\end{bmatrix}^\tp .
 \end{align*}
 Stacking the vectors together for all $k=1,\ldots,m$, i.e. 
 \begin{align*}
  \y^d_{\nabla} &= \begin{bmatrix} (\y^d_{\nabla,1})^\tp & \cdots & (\y^d_{\nabla,m})^\tp \end{bmatrix}^\tp,  \\
  \x &= \begin{bmatrix} \x_1^\tp & \cdots & \x_m^\tp \end{bmatrix}^\tp ,  
 \end{align*}
 $\u'$ is given by
 \begin{equation*}
  \u' = - \left( \P(\x-\y^d_\nabla) - \y^d_\delta \right) ,
 \end{equation*}
 and 
 \begin{equation} \label{eq:uORT_MIMO}
  \u = (\D^\star)^{-1} \u' = - (\D^\star)^{-1} \left( \P(\x-\y^d_\nabla) - \y^d_\delta \right) 
 \end{equation}
 which is identical to \eqref{eq:uFBT_MIMO}. 
\end{proof}

From the above prove, it is clear that Corollary~\ref{col:exo} is also true in the MIMO case. Furthermore, similar to Corollary~\ref{col:cancel_pz}, we have the following result for MIMO systems with $\delta<n$.

\begin{corollary}
 Given that $\D^\star$ is invertible, the invariant zeros can be compensated by closed-loop poles and an exosystem formulation can be omitted. 
\end{corollary}

\begin{proof}
 Let us define the matrix
 \begin{equation} \label{eq:T2}
  \T_n = \begin{bmatrix} \T \\ \t_{\delta+1} \\ \vdots \\ \t_n \end{bmatrix} \in \RZ^{n \times n} ,
 \end{equation}
 where $\T$ is determined as in \eqref{eq:T} and $\t_{\delta+1}, \ldots, \t_n$ are chosen such that $\T_n$ is regular. 
 For the transformed system $\A_n = \T_n\A\T_n^{-1}$, $\B_n = \T_n\B$, $\C_n = \C\T_n^{-1}$, we get 
 \begin{equation*}
  \A_n = \begin{bmatrix} \A' & \A_\delta \\ \A_\phi & \A_\eta \end{bmatrix} , \quad \B_n = \begin{bmatrix} \B' \\ \B_\eta \end{bmatrix}, \quad \C_n = \begin{bmatrix} \C' & \zero \end{bmatrix} ,
 \end{equation*}
 with $\A' \in \RZ^{\delta \times \delta}, \B' \in \RZ^{\delta \times m}, \C' \in \RZ^{m \times \delta}$ as before, and
 \begin{equation}
  \A_\delta = \begin{bmatrix} \A_{\delta,1} \\ \vdots \\ \A_{\delta,m} \end{bmatrix} \in \RZ^{\delta \times (n-\delta)} , \qquad  \A_{\delta,k} = \begin{bmatrix} 0 & \cdots & 0 \\ \vdots & & \vdots \\ 0 & \cdots & 0 \\ \star & \cdots & \star \end{bmatrix} \in \RZ^{\delta_k \times (n-\delta)} .
 \end{equation}
 The matrices $\A_\phi \in \RZ^{(n-\delta) \times \delta}$, $\A_\eta \in \RZ^{(n-\delta) \times (n-\delta)}$ and $\B_\eta \in \RZ^{(n-\delta) \times m}$ do not have any specific structure. 
 Let us denote the last row of $\A_{\delta,k}$ as $\A_{\delta,k}^\star$ and $\A_{\delta}^\star = \begin{bmatrix} (\A_{\delta,1}^\star)^\tp & \cdots & (\A_{\delta,m}^\star)^\tp \end{bmatrix}^\tp$. Then, provided that $\D^\star$ is invertible, we can design a feedback controller with gain $\K_n = (\D^\star)^{-1} \begin{bmatrix} \zero & \A_\delta^\star \end{bmatrix}$ such that the closed-loop system matrix reads
 \begin{equation}
  \A_n - \B_n\K_n = \begin{bmatrix} \A' & \zero \\ \A_\phi & \A'_\eta \end{bmatrix}, \quad \text{ with } \quad \A'_\eta = \A_\eta-\B_\eta(\D^\star)^{-1}\A_\delta^\star .
 \end{equation}
 The dynamics $\A'_\eta$ are not observable and represent the internal dynamics of the system. The reduced system $(\A',\B',\C')$ -- with dynamical order $\delta$ -- has no zero dynamics and Theorem~\ref{thm:ORT_FBT_MIMO} is directly applicable for it. This means that, like in the SISO case, we can compensate the zero dynamics by an appropriate feedback and design the tracking controller without an exosystem. 
\end{proof}

Note that the eigenvalues of the internal dynamics $\A'_\eta$ are the invariant zeros of the plant. If there are invariant zeros in the closed right half of the complex plane, the controller $\K_n$ leads to an unstable closed-loop behavior.

% \begin{remark}
 Additionally to the minimum phase characteristics of the plant, there is another restriction on the zero compensation approach in case of MIMO systems. For the compensation gain $\K_n$, it is required that $\D^\star$ is invertible. Thus, the compensation approach is only applicable for systems with regular $\D^\star$. 
% (which in the SISO case is always satisfied). 
%  However, in general it is not necessary to have a regular $\D^\star$ for the ORT approach.
% \end{remark}

\begin{remark}
 For $\delta = n$, $\D^\star$ is regular as shown in the proof of Theorem~\ref{thm:ORT_FBT_MIMO}, but for $\delta<n$ this must not be true.
\end{remark}

It should be noted that if the FBT approach is directly applied to systems with $\delta < n$, this results into the zero compensation method described above. Then, the error dynamics \eqref{eq:ey}--\eqref{eq:ey_compact} or the tracking controller \eqref{eq:uFBT_MIMO} leads to an input-output decoupled system for which the internal zero dynamics are made unobservable by pole-zero compensations. In this case, the regularity of $\D^\star$ is a necessary (and sufficient) condition for the zero compensation or input-output decoupling (cf. e.g. \citep{FalbWolovich1967}), and leads to unstable closed-loop dynamics for non-minimum phase systems.
However, the ORT approach enables tracking even if the system is non-minimum phase or if $\D^\star$ is singular. The only requirement is that the plant is stabilizable and that no invariant zero coincides with an eigenvalue of the exosystem. It is not necessary to compensate the zero dynamics, but then the reference trajectory must be provided by an exosystem.

% ===============================================
\section{Conclusion} \label{sec:conclusion}
% ===============================================

Two different approaches for the trajectory tracking problem of multivariable linear systems have been considered: the \emph{output regulation theory} (ORT) and the \emph{flatness based tracking} (FBT). It has been shown that both methods lead to identical tracking controllers if the output of the plant is flat. Moreover, an analytic solution of the \emph{regulator equations} has been derived which sheds new light into the tracking controller of the ORT. Based on this solution, the connection to the FBT approach has been discussed and it has been shown that an \emph{exosystem} formulation for the reference trajectory can be omitted if the controller parameters are chosen appropriately. In case of non-minimum phase systems, however, this special choice leads to an unstable closed-loop behavior. 

An interesting point for future research would be to extend the results of this paper to nonlinear systems. One of the main difficulties is to find an analytic solution of the regulator equations which in the nonlinear case are partial differential equations. Presumably, a solution in closed form may be impossible to find.

% % ===============================================
% \section*{Acknowledgements}
% % ===============================================
% 
% ... 
% 
% 
% 
% % ===============================================
% \section*{Funding}
% % ===============================================
% 
% The author gratefully acknowledges support from the German Research Foundation (DFG) within the GRK 1362 (\texttt{www.gkmm.de}).

% ===============================================
% References
% ===============================================

% ===============================================
\appendices
% ===============================================

% ----------------------------------------
\section{Proof of Proposition~\ref{prop:Mk}} \label{app:prop:Mk}

Let us define the matrices $\mXi_k \in \RZ^{(n-m) r_k \times n r_k}$ and $\mPsi_k \in \RZ^{n r_k \times (n-m) r_k}$:
\allowdisplaybreaks
\begin{align*} 
 \mXi_k &= \begin{bmatrix} \begin{matrix} \zero & -\I_{r_k} & \cdots & \zero \\ \vdots & \vdots & \ddots & \vdots \\ \zero & \zero & \cdots & -\I_{r_k} \end{matrix}\rlap {\hspace*{.8mm} $\left.\vbox to 26pt{} \right\} (\kappa_1-1) \cdot r_k \, $ rows} & & \\ & \ddots & \\ & & \rlap{\hspace*{-39mm} $(\kappa_m-1) \cdot r_k \, \text{ rows }\left\{ \vbox to 26pt{} \right.$} \begin{matrix} \zero & -\I_{r_k} & \cdots & \zero \\ \vdots & \vdots & \ddots & \vdots \\ \zero & \zero & \cdots & -\I_{r_k} \end{matrix} \end{bmatrix} \\ 
 \mPsi_k &= \begin{bmatrix} \begin{matrix} \S_k & \zero & \cdots & \zero \\ -\I_{r_k} & \S_k & \cdots & \zero \\ \zero & -\I_{r_k} & \ddots & \vdots \\ \vdots & \vdots & \ddots & \S_k \\ \zero & \zero & \cdots & -\I_{r_k} \end{matrix}\rlap {\hspace*{2mm}$\left.\vbox to 45pt{} \right\} \kappa_1 \cdot r_k \, $ rows} &  & \\ & \ddots & \\ &  & \rlap{\hspace*{-28mm}$\kappa_m \cdot r_k \, \text{ rows }\left\{ \vbox to 45pt{} \right.$}\begin{matrix} \S_k & \zero & \cdots & \zero \\ -\I_{r_k} & \S_k & \cdots & \zero \\ \zero & -\I_{r_k} & \ddots & \vdots \\ \vdots & \vdots & \ddots & \S_k \\ \zero & \zero & \cdots & -\I_{r_k} \end{matrix} \end{bmatrix} 
\end{align*}
where the empty entries are all zero. Then, the following equation holds
\begin{equation} \label{eq:aux}
  \underbrace{\begin{bmatrix} \mXi_k \\ \overline{\S}_k \end{bmatrix}}_{\mTheta_k} \underbrace{\begin{bmatrix} \mPsi_k & \overline{\C}^\tp_k \end{bmatrix}}_{\Z_k} = \begin{bmatrix} \mPhi_k & \mXi_k \overline{C}^\tp_k \\ \zero & \M_k \end{bmatrix}
\end{equation}
with
\begin{equation*}
  \mPhi_k = \mXi_k \mPsi_k = \begin{bmatrix} \I_{r_k} & -\S_k & \cdots & \zero \\ \zero & \I_{r_k} & \ddots & \vdots \\ \vdots & \vdots & \ddots & -\S_k \\ \zero & \cdots & \cdots & \I_{r_k} \end{bmatrix} \in \RZ^{(n-m) r_k \times (n-m) r_k} .
\end{equation*}
Note that $\mTheta_k$ on the left-hand side of \eqref{eq:aux} is regular since there is exactly one identity matrix in every column and row block. Therefore, the regularity of the matrix on the right-hand side of \eqref{eq:aux} depends on the regularity of $\Z_k$. Obviously, $\mPhi_k$ is a regular matrix, thus, $\M_k$ is regular/singular if $\Z_k$ is regular/singular. However, $\Z_k$ is regular if and only if no invariant system zero coincides with an exosystem eigenvalue. To see this, consider the \emph{Rosenbrock's} system matrix
\begin{equation} \label{eq:RccfMIMO}
  \R(\lambda) = \left[ \begin{array}{ccccccccc|ccc} \lambda & -1 & \cdots & 0 & & & & & & 0 & & \\ 0 & \lambda & \ddots & \vdots & & & & & & 0 & & \\ \vdots & \vdots & \ddots & -1 &  & & & & & 0 & & \\ 0 & 0 & \cdots & \lambda & & & & & & -1 & & \\  &  &  &  & \ddots &  &  &  &  &  & \ddots &  \\ & & & & & \lambda & -1 & \cdots & 0 & & & 0 \\ & & & & & 0 & \lambda & \ddots & \vdots & & & 0 \\ & & & & & \vdots & \vdots & \ddots & -1 & & & \vdots \\ & & & & & 0 & 0 & \cdots & \lambda & & & -1 \\ \hline &&&& \C &&&& &&\zero& \end{array} \right] .
\end{equation}
Expanding from the last $m$ columns, it is clear that $\rank{\R(\lambda)} = \rank{\overline{\R}(\lambda)}$, where $\overline{\R}(\lambda)$ is the matrix which results from deleting the last $m$ columns and the rows $\sum_{i=1}^j \kappa_i$ for all $j=1,\ldots,m$: %in $\R(\lambda)$:
\begin{equation} \label{eq:R-ccfMIMO}
  \overline{\R}(\lambda) = \left[ \begin{array}{ccccccccccc} \lambda & -1 & 0 & \cdots & 0 & & & & & & \\ 0 & \lambda & -1 & \cdots & 0 & & & & & & \\ \vdots & \vdots & \ddots & \ddots & \vdots & & & & & & \\ 0 & 0 & \cdots & \lambda & -1 & & & & & & \\ & & & & & \ddots & & & & & \\ & & & & & & \lambda & -1 & 0 & \cdots & 0 \\ & & & & & & 0 & \lambda & -1 & \cdots & 0 \\ & & & & & & \vdots & \vdots & \ddots & \ddots & \vdots \\ & & & & & & 0 & 0 & \cdots & \lambda & -1 \\ \hline &&&&& \C &&&&& \end{array} \right] .
\end{equation}
Note that $\overline{\R}(\lambda)$ is singular for every $\lambda$ which is an invariant zero of the system or, equivalently, the invariant zeros are the roots of the polynomial $\overline{P}(\lambda) = \det{\overline{\R}(\lambda)}$. Comparing \eqref{eq:R-ccfMIMO} with $\Z_k$, it can be seen that $\Z_k^\tp$ is equal to $(\overline{\R}(\lambda) \otimes \I_{r_k})$ if $\lambda$ is replaced by $\S_k$, i.e. $ \Z_k^\tp = (\overline{\R}(\lambda) \otimes \I_{r_k})|_{\lambda=\S_k}$. Using determinant and \emph{Kronecker} product rules, we know that the following identity holds: $\det{\overline{\R}(\lambda) \otimes \I_{r_k}} = \det{\overline{\R}(\lambda)} \det{\I_{r_k}} = \det{\overline{\R}(\lambda)} = \overline{P}(\lambda)$. Thus, by the \emph{Cayley-Hamilton} theorem, the matrix $\overline{\P}(\S_k)=(\overline{P}(\lambda)\otimes \I_{r_k})|_{\lambda=\S_k}$ is regular if and only if no invariant zero of the plant coincides with an eigenvalue of $\S_k$. This implies that $\Z_k$ is regular and since regularity of $\M_
k$ is equivalent to regularity of $\Z_k$, Proposition~\ref{prop:Mk} is proved.

% ----------------------------------------
\section{Rank of transformation matrix \textbf{\textit{T}}} \label{app:T}

It is shown that the rank of the matrix $\T$ in \eqref{eq:T} is $\delta$. 
The proof is inspired by \citep{Lohmann1991d} and given as follows. Consider the equation 
\begin{equation} \label{eq:dof}
 \begin{bmatrix} \lambda_{ij}\I_n-\A & -\B \\ \C & \zero \end{bmatrix} \begin{bmatrix} \v_{ij} \\ \w_{ij} \end{bmatrix} = \begin{bmatrix} \zero \\ \e_{i} \end{bmatrix}, \quad i=1,\ldots,m, \quad j=1,\ldots,\delta_i ,
\end{equation}
where $\e_i$ is the unit vector with the $i$-th element equal to $1$ and all other elements $0$. The matrix on the left-hand side of equation \eqref{eq:dof} is the \emph{Rosenbrock's} system matrix. Thus, choosing $\lambda_{ij}$ different from the invariant zeros, \eqref{eq:dof} can be solved for $\v_{ij}$ and $\w_{ij}$. In this case, the second row block of \eqref{eq:dof} yields 
\begin{equation} \label{eq:dof1} 
 \C_k \v_{ij} = \begin{cases} 1 ,& k=i, \\ 0 ,& k \ne i . \end{cases}
\end{equation}
Multiplying the first row block of \eqref{eq:dof} with $\C_k, \ldots, \C_k\A^{\delta_k-2}$ leads to
\begin{subequations} \label{eq:dof2} 
\begin{align} 
 \C_k \A \v_{ij} + \underbrace{\C_k \B}_{\zero} \w_{ij} &= \begin{cases} \lambda_{ij} ,& k=i, \\ 0 ,& k \ne i , \end{cases}  \\
 & \hspace{2mm} \vdots \notag \\
 \C_k \A^{\delta_k-1} \v_{ij} + \underbrace{\C_k \A^{\delta_k-2} \B}_{\zero} \w_{ij} &= \begin{cases} \lambda_{ij}^{\delta_i-1} ,& k=i, \\ 0 ,& k \ne i . \end{cases}
\end{align}
\end{subequations}
Eqs. \eqref{eq:dof1} and \eqref{eq:dof2} are written in a more compact way as 
\begin{equation*}
 \begin{bmatrix} \C_k \\ \vdots \\ \C_k \A^{\delta_k-1} \end{bmatrix} \v_{ij} = \begin{cases} \mLambda_{ij} ,& k=i, \\ \zero ,& k \ne i , \end{cases} \qquad \text{ with } \quad \mLambda_{ij} = \begin{bmatrix} 1 \\ \lambda_{ij} \\ \vdots \\ \lambda_{ij}^{\delta_i-1} \end{bmatrix} .
\end{equation*}
This procedure is repeated for all $j=1,\ldots,\delta_i$, leading to
\begin{equation} \label{eq:TkVi}
 \begin{bmatrix} \C_k \\ \vdots \\ \C_k \A^{\delta_k-1} \end{bmatrix} \underbrace{\begin{bmatrix} \v_{i1} & \cdots & \v_{i\delta_i} \end{bmatrix}}_{\V_i} = \begin{cases} \mLambda_{i} ,& k=i, \\ \zero ,& k \ne i , \end{cases} 
\end{equation}
with $\mLambda_{i} = \begin{bmatrix} \mLambda_{i1} & \cdots & \mLambda_{i\delta_i} \end{bmatrix}$. Writing \eqref{eq:TkVi} for all $k,i \in \{1,\ldots,m\}$ in a matrix, it can easily be seen that 
\begin{equation} \label{eq:TVL}
 \T \underbrace{\begin{bmatrix} \V_1 & \cdots & \V_m \end{bmatrix}}_{\V \in \RZ^{n \times \delta}} = \begin{bmatrix} \mLambda_1 & \cdots & \zero \\ \vdots & \ddots & \vdots \\ \zero & \cdots & \mLambda_m \end{bmatrix} \in \RZ^{\delta \times \delta} ,
\end{equation}
 where $\T$ is given as in \eqref{eq:T}. Note that the matrices 
\begin{equation*} 
 \mLambda_{i} = \begin{bmatrix} 1 & \cdots & 1 \\ \lambda_{i1} & \cdots & \lambda_{i\delta_i} \\ \vdots &  & \vdots \\ \lambda_{i1}^{\delta_i-1} & \cdots & \lambda_{i\delta_i}^{\delta_i-1} \end{bmatrix} \in \RZ^{\delta_i \times \delta_i}
\end{equation*}
are \emph{Vandermonde} matrices which are regular if all the numbers $\lambda_{ij}$ are distinct. Thus, for distinct $\lambda_{ij}$, the matrix on the right-hand side of \eqref{eq:TVL} is regular. But this is only possible if the rank of both matrices $\T$ and $\V$ on the left-hand side of \eqref{eq:TVL} is $\delta$.

% ===============================================
\end{document}